\documentclass[a4paper, conference]{IEEEtran} %{a4paper}
\usepackage{graphicx}
\usepackage{pdfsync}

\newtheorem{theorem}{Theorem}%[section]

\newtheorem{example}{Example} %[theorem]{Example}

\ifCLASSINFOpdf
\else
\fi
\begin{document}
\title{Recovery of bilevel causal signals with finite rate of innovation  using positive  sampling kernels}


\author{\IEEEauthorblockN{Gayatri Ramesh}
\IEEEauthorblockA{Department of Mathematics\\
University of Central Florida\\
Orlando, FL 32816\\
Email: gayatriramesh@knights.ucf.edu}
\and
\IEEEauthorblockN{Elie Atallah  }
\IEEEauthorblockA{Department of Mathematics\\
University of Central Florida\\
Orlando, FL 32816\\
Email: elieatallah@knights.ucf.edu}
\and
%\IEEEauthorblockN{Ram N. Mohapatra}
%\IEEEauthorblockA{Department of Mathematics\\
%University of Central Florida\\
%Orlando, FL 32816\\
%Email: ram.mohapatra@ucf.edu}
%\and
\IEEEauthorblockN{Qiyu Sun}
\IEEEauthorblockA{Department of Mathematics\\
University of Central Florida\\
Orlando, FL 32816\\
Email: qiyu.sun@ucf.edu}
}


%Telephone: (800) 555--1212, Fax: (888) 555--1212}


\maketitle

\begin{abstract}
%\boldmath
Bilevel signal $x$ with maximal local rate  of innovation $R$ is a continuous-time signal that
takes only two values $0$ and $1$ and that there is at most one transition position in any time period of  $1/R$.
In this note, we introduce a  recovery method for
 bilevel causal signals $x$ with maximal local rate  of innovation $R$
 from their uniform samples $x*h(nT), n\ge 1$, where the sampling kernel $h$ is causal and positive on  $(0, T)$, and
 the sampling rate $\tau:=1/T$ is at (or above) the maximal local rate  of innovation $R$.
  We also discuss
   stability of the bilevel signal recovery procedure in the presence of bounded noises.
\end{abstract}


% \fi
\IEEEpeerreviewmaketitle

\section{Introduction}

Let $T>0$ and $N$ be a nonnegative integer  or infinity, and denote by $\chi_E$ the indicator function on a set $E$.
In this note, we consider  {\em bilevel causal signals} %$x$  with finite rate of innovation,
\begin{equation}\label{bilevel.def} x(t):=\sum_{i=1}^{N} \chi_{[t_{2i-1}, t_{2i})}(t)\end{equation}
with unknown  transition values (positions)
$t_i, 1\le i\le 2N$, satisfying
\begin{equation}\label{innovationassump} t_i< t_{i+1},  \ 1\le i<2N;
\end{equation} % being determined,
and also
a uniform generalized sampling process
\begin{equation}\label{sampling.def} x(t)\longmapsto x*h(t)\longmapsto \{x*h(nT)\}_{n\ge 1}\end{equation}
with  sampling kernel $h$ being causal  and   uniform sampling taken  every $T$ seconds.
 For the bilevel  causal signal $x$ in (\ref{bilevel.def}), define its {\em maximal  local rate of innovation}
$R$ by  reciprocal of the maximal positive  number $\sigma_0$ such that there is at most one transition position $t_i, 1\le i\le 2N$,
in any time
period $[t, t+\sigma_0), t\ge 0$, that is,
\begin{equation}\label{innovativerate.def}
R=\sup_{1\le i<2N} \frac{1}{t_{i+1}-t_i}.\end{equation}
The concept of signals with finite rate of innovation was introduced by Vetterli, Marziliano and  Blu
\cite{vetterli02}. Examples of signals with finite rate of innovation
include streams of Diracs, piecewise polynomials, band-limited signals, and signals in a finitely-generated
shift-invariant space \cite{vetterli02}--\cite{sunaicm08}.
In  the past ten years, the paradigm for reconstructing %esenting % band-limited signals
%in the classical Shannon's sampling theory \cite{shannon49} and signals in shift-invariant spaces \cite{aldroubi, unser} was extended to represent
signals with finite rate of innovation from their samples has been developed, see for instance \cite{vetterli02}, \cite{win02} and \cite{sunaicm08}--\cite{sun13} and references therein.

Precise identification of transition  positions is important to reach meaningful conclusions
in many applications.
%, such as global positioning system, cellular radio, ultra wide-band
%communication, and mass spectrometry.
Vetterli, Marziliano and  Blu show in \cite{vetterli02} that  a bilevel  signal $x$  defined in (\ref{bilevel.def}) can be reconstructed
from its samples (\ref{sampling.def}) when the sampling kernel $h$ is the box spline  $\chi_{[0,T)}$
(or the hat spline $ (T-|t|)\chi_{[-T, T)}(t)$) and  the sample rate $\tau:=1/T$ is at (or above) the maximal local rate of innovation $R$ of the signal $x$.
In this note, we show that  bilevel causal signals $x$ defined
in (\ref{bilevel.def}) are uniquely determined from their samples $x*h(nT), n\ge 1$, in (\ref{sampling.def})
if the  sampling kernel $h$ is  causal and positive on $(0,T)$,  and the sample rate $\tau$ is at (or above) the maximal local rate of innovation $R$, see Theorem \ref{uniqueness.tm}. Our numerical simulations indicate that    the bilevel signal recovery procedure from noisy samples $x*h(nT)+\epsilon_n, n\ge 1,$ is stable
 when  there are limited numbers of transition positions for the bilevel signal $x$.
 % the noise level $\epsilon=\max_{n\ge 1} |\epsilon_n|$
%is  at (or below) $2\%$ of the maximal sample value
%$\max_{n\ge 1} |x*h(n/T)|$.

\section{Recovery of bilevel causal signals}
\label{recovery.section}

In this section, we provide a necessary condition on the sampling kernel $h$
such that   bilevel signals $x$ in (\ref{bilevel.def})  are uniquely determined from their samples $\{x*h(nT)\}$
in (\ref{sampling.def}).
Also in this section, we propose an algorithm for the bilevel signal recovery.

\smallskip
%\subsection{Main theorem}

The main  theorem of this note is as follows:

%\smallskip

\begin{theorem}\label{uniqueness.tm} Let $T>0$ and set $\tau=1/T$.
If  $h$ is a causal sampling kernel with $h(t)>0$ on $(0, T)$, then any bilevel causal signal $x$
in (\ref{bilevel.def})  with maximal local rate of innovation $R$ being less than or equal to the sampling rate $\tau$ can be recovered
 from its samples $x*h(nT), n\ge 1$. %, in (\ref{sampling.def}).
\end{theorem}

\smallskip

\begin{proof}
Let
\begin{equation} H(t)=\int_0^t h(s) ds,\ 0\le t\le T.
\end{equation}
Then $H(0)=0$ and  $H$ is a strictly increasing function on $[0,T)$ as $h$ is strictly positive on $(0, T)$. Denote its inverse function
on $[0, T]$
by $H^{-1}:[0, H(T)]\longmapsto [0, T]$.

Let $x$ be a bilevel causal signal in (\ref{bilevel.def}) with transition positions $t_i, 1\le i\le 2N$, satisfying
 (\ref{innovationassump}).
Then its first sample $y_1= x*h(T)$ is given by
\begin{eqnarray*}
y_1 & = & \int_0^\infty x(t) h\big(T-t\big) dt=\int_0^{T} x(t) h\big(T-t\big) dt\\
&= & \int_{0}^{T} \chi_{[t_1, t_2)}(t) h\big(T-t\big) dt= H\big(\max\{T-t_1, 0\big\}\big),
\end{eqnarray*}
where the first two equalities hold by the causality of the signal $x$ and the sampling kernel $h$,
and the fourth equality follows from (\ref{bilevel.def}) and the observation
that
$$t_i\ge t_2=(t_2-t_1)+t_1\ge 1/R+0 \ge 1/\tau=T, \ i\ge 2$$ by  (\ref{innovationassump}), (\ref{innovativerate.def})
and the assumption that $R\le \tau$.
Recall that
 $H$ is strictly increasing on $[0, T)$. Then
 there exists a transition position in the time range $[0,T)$
 if and only if  $y_1=x*h(T)>0$. Moreover, if it exists,  it is given  by
 \begin{equation} t_1=T- H^{-1}(y_1).\end{equation}
 %where $H^{-1}$ is the inverse function of the function $H$ on $[0, T)$.
 Thus  for a bilevel causal signal, we may determine from its first sample $x*h(T)$
 the (non-)existence of its transition position in the time period $[0, T)$ and
  further its transition  position in that time period if there is one.

 Inductively, we  assume that all transition positions of the bilevel signal $x$ in the time range $[0, nT)$
 have been determined from its
 samples $y_k=x*h(kT), 1\le k\le n$. We examine four cases to determine its transition position in the time period $[nT, (n+1)T)$ from the
 sample $y_{n+1}=x*h((n+1)T)$.

 {\bf Case 1}: \ There is no transition position in $[0, nT)$.

 In this case, following the above argument to determine transition positions in the time range $[0, T)$, we have that
  there exists a transition position in $[nT, (n+1)T)$ if and only if $y_{n+1}>0$. If there is, the transition position
 is the first transition position $t_1$ of the bilevel causal signal $x$,  and
 %its transition value $t_1$ is given by
 \begin{equation} t_1=(n+1)T- H^{-1}(y_{n+1}).\end{equation}

 {\bf Case 2}: \ The last transition position in  $[0, nT)$ is $t_{2i_0-1}$ for some $i_0\ge 1$.

 In this case,
 $t_{2i_0}\ge nT$ and $t_i\ge (n+1)T$ for all $i>2i_0$.
 Thus  %the $(n+1)$-th sample $y_{n+1}=x*h((n+1)/T)$ is obtained by
 \begin{eqnarray*}
 y_{n+1} & = & %x*h\big((n+1)T\big)=
 \int_0^{(n+1)T} x(t) h\big((n+1)T-t\big) dt\\
  & = & \int_0^{(n+1)T} h\big((n+1)T-t\big)\nonumber\\
  & &\quad \times
   \Big(\sum_{i=1}^{i_0-1} \chi_{[t_{2i-1}, t_{2i})}(t) +\chi_{[t_{2i_0-1}, (n+1)T)}(t)\Big)
    dt \nonumber\\
    & &   -
   \int_{nT}^{(n+1)T}
    h\big((n+1)T-t\big)
   \\
   & &\qquad
   \times     \chi_{[\min(t_{2i_0}, (n+1)T), (n+1)T) }(t) dt.
%  & & \quad + H(T)-
% H(\max((k+1)T-t_{2i+1},0)).
\end{eqnarray*}
% Recall that $\int_0^{kT} x(t) h((k+1)T-t) dt$ is determined by the innovation locations in the time range $[0, kT)$
% by the inductive hypothesis.
Hence there exists a transition position $t_{2i_0}$ in the time range $[nT, (n+1)T)$ if and only if
 \begin{eqnarray}\label{yn+1.case2}
   & & \tilde  y_{n+1} :=  -y_{n+1}+ \int_0^{(n+1)T} h\big((n+1)T-t\big)\nonumber\\
  & & \quad  \times
   \Big(\sum_{i=1}^{i_0-1} \chi_{[t_{2i-1}, t_{2i})}(t) +\chi_{[t_{2i_0-1}, (n+1)T)}(t)\Big)
    dt\quad \end{eqnarray}
 is positive. Moreover if $\tilde y_{n+1}>0$,
 the transition position $t_{2i_0}$ in the time range $[nT, (n+1)T)$ is determined by
 \begin{equation}\label{ti.case2}
 t_{2i_0}= (n+1)T- H^{-1}(\tilde y_{k+1}).
 \end{equation}
% where $\tilde y_{k+1}=H(T)+\int_0^{kT} x(t) h((k+1)T-t) dt-y_{k+1}$.

  {\bf Case 3}: \ The last transition position in  $[0, nT)$ is $t_{2i_0}$ for some $1\le i_0<N$.

In this case,  the $(n+1)$-th sample $y_{n+1}=x*h((n+1)T)$ is given by
 \begin{eqnarray}
 y_{n+1} & = & \int_0^{nT}\Big (\sum_{i=1}^{i_0}
 \chi_{[t_{2i-1}, t_{2i})(t)} \Big) h\big((n+1)T-t\big) dt\nonumber\\
 & & +\int_{\min(t_{2i_0+1}, (n+1)T)}^{(n+1)T} h\big((n+1)T-t\big) dt.
 %\\
%  & = & \int_0^{kT} x(t) h((k+1)T-t) dt\nonumber\\
%  & & \quad +
% H(\max((k+1)T-t_{2i},0)).
\end{eqnarray}
 Thus there exists a transition value $t_{2i_0+1}\in [nT, (n+1)T)$ if and only if
 \begin{equation}
 \label{yn+1.case3}
 \tilde y_{n+1}:=y_{n+1}-\int_0^{nT}\Big (\sum_{i=1}^{i_0}
 \chi_{[t_{2i-1}, t_{2i})(t)} \Big) h\big((n+1)T-t\big) dt\end{equation} is positive.
 Also we see that if $\tilde y_{n+1}>0$, then  the transition value $t_{2i_0+1}$ can be obtained by
 \begin{equation}\label{ti.case3}
 t_{2i_0+1}= (n+1)T- H^{-1}(\tilde y_{n+1}).
 \end{equation}

  {\bf Case 4}: \ The last transition position in $[0, nT)$ is $t_{2N}$.

  In this case, all transition positions of the bilevel signal $x$ have been  recovered already. Hence the bilevel signal $x$ is fully recovered.

This completes our inductive proof.
\end{proof}

\smallskip

%\subsection{Algorithm to recover bilevel signals form their samples}
From the above argument of Theorem \ref{uniqueness.tm}, we can use the following algorithm to recover a bilevel causal signal $x$
in (\ref{bilevel.def})  from its samples $x*h(nT), 1\le n\le K$,% in (\ref{sampling.def}), % in the noiseless environment,
 where $K>t_{2N}\tau$:

{\bf Bilevel Signal Recovery Algorithm}:
\begin{quote}

\begin{itemize}%[xxxxxx]
\item[{\em Step 1}:] If all samples $y_n=x*h(nT), 1\le n\le K$, are zero, then set $x=0$ and stop;
 else
find the first nonzero sample, say $y_{n_0}>0$,  the first transition position of the bilevel signal $x$ is located at
$t_1:= n_0-H^{-1}(y_{n_0})$, and set $n=n_0$.

\item [{\em  Step 2}:] Do Step 2a if
 the last transition position in the time range $[0, nT)$ is $t_{2i_0-1}$ for some $i_0\ge 1$; do
Step 2b elseif  the last transition position in the time range $[0, nT)$ is $t_{2i_0}$ for some $1\le i_0<N$;
and do
Step  4 else.

\begin{itemize}

\item Step 2a:  Define $t_{2i_0}$ as in (\ref{ti.case2}) if  $\tilde y_{n+1}$   in (\ref{yn+1.case2})
is positive, else do Step 3.
%  set $i=i+1$
% if
%$\tilde y_{n+1}>0$. %, else set $n=n$.

  \item Step 2b:  Define $t_{2i_0+1}$ as in (\ref{ti.case3}) if
   $\tilde y_{n+1}$ in (\ref{yn+1.case3}) is positive, else do Step 3.

  \end{itemize}

  \item[{\em  Step 3}:] Set $n=n+1$. Do Step 2 if $n<K$, and Step 4 if $n=K$.

\item[{\em Step 4}:]
 Stop as all transition positions $t_i, 1\le i\le 2N$, of the bilevel signal $x$ are recovered.
 \end{itemize}
\end{quote}

\smallskip
We finish this section with a  remark that  the requirement $R\le \tau$ in Theorem \ref{uniqueness.tm}
 can be relaxed to
the following: There is at most one transition position $t_i, 1\le i\le 2N$, in each sampling range $[nT, (n+1)T), n \ge 1$.

\section{Stable recovery of bilevel causal signals}

In this section, we consider
the maximal sampling error $\sup_{n} |x*h(nT)-\tilde x*h(nT)|$
of two bilevel signals $x$ and $\tilde x$ when maximal error of
their transition positions are small.
We then present some numerical simulations on recovery of a bilevel signal
$x$ in (\ref{bilevel.def}) from its noisy samples $\{x*h(nT)+\epsilon_n\}$ in (\ref{sampling.def}),
where $\epsilon_n, n\ge 1$, are bounded noises of low levels.

\smallskip

First we notice that sampling procedure from bilevel signals $x$ to their samples $\{x*h(nT)\}$ are stable in bounded norm.

%\smallskip
%
%\subsection{Stability of sampling procedure in the bounded norm}

\begin{theorem}\label{stability.tm}
Let $T>0$,  $h$ be a  bounded filter supported in $[0, MT)$,  $x(t)=\sum_{i=1}^N \chi_{[t_{2i-1}, t_{2i})}(t)$ be a bilevel causal signal
with maximal local innovation rate $R\le \tau:=1/T$, and
$\tilde x(t)=\sum_{i=1}^N \chi_{[t_{2i-1}+\delta_{2i-1}, t_{2i}+\delta_{2i})}$
be a perturbation of the bilevel signal $x$ with perturbed  transition positions $\{t_i+\delta_i\}_{i=1}^{2N}$  satisfying
$$\delta:=\sup_{1\le i\le 2N} |\tilde t_i-t_i|<\frac{1}{2R}.$$
Then the sample errors between $x*h(nT)$ and $\tilde x*h(nT), n\ge 1$, are dominated by
$(\lfloor MRT\rfloor +2)  \|h\|_\infty \delta$, i.e.,
\begin{equation}\label{error.prop.eq1}
|x*h(nT)-\tilde x*h(nT)|\le
\big(\lfloor MRT\rfloor +2\big)  \|h\|_\infty \delta,\  n\ge 1,
\end{equation}
where $\|h\|_\infty$ is the $L^\infty$ norm of the sampling kernel $h$.
\end{theorem}

\begin{proof}
By the assumption on maximal local innovation rate $R$
of the bilevel signal $x$ and  the  maximal transition position  perturbation $\delta$ between bilevel signals $x$ and $\tilde x$, we have that
%\begin{equation*}%\label{xtildexdifference}
$$ |x(t)-\tilde x(t)|=\sum_{i=1}^{2N} \chi_{t_i+[\min(\delta_i,0), \max(\delta_i,0))}(t).$$ %\end{equation*}
This together with the support assumption for the sampling kernel $h$
gives that
\begin{eqnarray*}
& & |x*h(nT)-\tilde x*h(nT)|\nonumber\\
 & = &
\Big|\int_0^{nT} (x(t)-\tilde x(t)) h(kT-t) dt\Big|\nonumber\\
& \le & \|h\|_\infty \int_{(n-M)T}^{nT}
\sum_{i=1}^{2N} \chi_{t_i+[\min(\delta_i,0), \max(\delta_i,0))}(t)
 dt.
\end{eqnarray*}
Therefore
\begin{eqnarray*}
& & |x*h(nT)-\tilde x*h(nT)|\nonumber\\
& \le &
\delta \|h\|_\infty
\#\{ t_i: t_i\in [(n-M)T-\delta, nT+\delta)\}
\nonumber\\
& \le  &
\delta \|h\|_\infty
(\lfloor (MT+2\delta)/(1/R)\rfloor +1)\nonumber \\
& \le  &
\delta \|h\|_\infty
(\lfloor M RT\rfloor +2),
\end{eqnarray*}
where the first inequality holds
as $t_i\in [(n-M)T-\delta, nT+\delta)$ if $t_i+[\min(\delta_i,0), \max(\delta_i,0))$
and $[(n-M)T, nT)$ have nonempty intersection, the second inequality is true as $t_{i+1}-t_i\ge 1/R$ for all $1\le i<2N$, and
the last inequality follows from the assumptions that $\delta<1/(2R)$ and $R\le \tau$.
This proves the sampling error estimate
(\ref{error.prop.eq1}) between the bilevel causal signals $x$ and $\tilde x$.
\end{proof}

\smallskip

Now we  consider the  corresponding nonlinear inverse problem how to recover a bilevel signal
$x$  from its noisy samples $\{x*h(nT)+\epsilon_n\}$ in (\ref{sampling.def}),
where $\epsilon_n, n\ge 1$, are bounded  noises.
Let us start by looking at two examples.

\smallskip

\begin{example}\label{example1}
Take $x_1(t)=\sum_{i=1}^\infty \chi_{[2i-1, 2i)}(t)$ as the original bilevel signal and $h_1(t)=\chi_{[0,2)}(t)$ as the sampling kernel.
For sufficiently small $\epsilon>0$, define $x_{1, \epsilon}
=\sum_{i=1}^\infty \chi_{[(1+\epsilon)(2i-1), 2 (1+\epsilon) i )}(t)$.
Then  for every $i\ge 1$, the $i$-th transition positions of bilevel signals $x_1$ and $x_{1, \epsilon}$ are $i$ and $i(1+\epsilon)$ respectively (hence their difference is $i\epsilon$ that could be arbitrary large for sufficiently large $i$),
but on the other hand,
maximal sampling errors for those two bilevel signals $x_1$ and $x_{1, \epsilon}$ are bounded by $\epsilon$,
$$|x_{1, \epsilon}*h_1(n)-x_1*h_1(n)|=|x_{1, \epsilon}*h_1(n)-1|\le \epsilon,\  n\ge 1.$$
This leads to  instability of the recovery procedure from
 samples $\{x_1*h_1(n)\}$  to the bilevel signal $x_1$ in the presence of bounded noises.
%
% a bilevel signal
%$x$  from its noisy samples $\{x*h(n/T)+\epsilon_n\}$ in (\ref{sampling.def}),
%where $\epsilon_n, n\ge 1$, are bounded  noises.
%
%even with small bounded sampling error,
%the difference  between innovation positions of bilevel signals $x_1$ and $x_{1, \epsilon}$ could be large.
\end{example}

\begin{example}\label{example2}
Take $x_1$ and $h_1$ in Example \ref{example1}
as the original bilevel signal and  the sampling kernel respectively. Define
$x_{2, \epsilon}=\sum_{i=1}^\infty \chi_{[2i-1+\epsilon, 2i+\epsilon)}(t) $
for sufficiently small $\epsilon>0$.
Then for every $i\ge 1$ the difference between  $i$-th transition positions of bilevel signals $x_1$ and $x_{2, \epsilon}$ is  always $\epsilon$,
and  there is no difference between their $n$-th samples except for $n=1$.
This  suggests that the recovery procedure from
 samples $\{x_1*h_1(n)\}$  to the bilevel signal $x_1$
 is not locally-behaved and the
 reconstruction error on transition positions could disseminate.
\end{example}

\smallskip
From the above two examples, we see that the nonlinear recovery procedure from
samples $\{x*h(n)\}$ to  bilevel signals $x$ is {\em unstable} in the presence of bounded noises and
that it is  {\em  globally-behaved} in general.
  In this note, we present some initial numerical simulations with
small numbers of transition positions, sampling rate over  maximal local rate of innovation
and very low levels of noise.
Detailed noise performance analysis and stable recovery in the presence of other types of noises will be discussed in the
 coming paper.

%Here are simple numerical simulations to demonstrate
%stability of the sampling procedure $x\longmapsto \{x*h(n/T)\}$. % before we shift to recovering bilevel signals from their noise samples.
Take  a sampling kernel
$h_0(t)=\frac{t+1}{2}\chi_{[0,1)}(t)+ (2t-1)\chi_{[1, 2)}$, and
 a bilevel signal
 \begin{eqnarray} \label{testsignal}
 x_0(t) &= & \chi_{[ 0.3791,    1.9885)}(t) +\chi_{[3.1306,    4.3440)}(t)\nonumber \\
 & & +\chi_{[5.7552,    7.1820)}(t)+
\chi_{[ 8.7423,  10.1052)}(t)\nonumber\\
 && +\chi_{[11.4200,   12.6884)}(t)\end{eqnarray}
  %0.3791    1.9885    3.1306    4.3440    5.7552    7.1820    8.7423   10.1052    11.4200   12.6884
containing 10 transition positions, see Figure \ref{bilevelsignal.fig}.
   \begin{figure}[hbt]
\centering
\begin{tabular}{cc}
  \includegraphics[width=43mm, height=32mm]{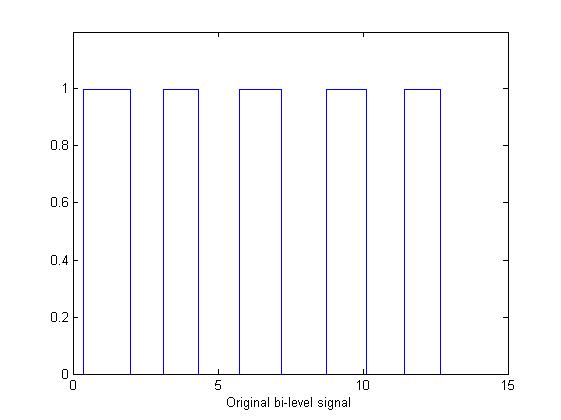} &
  \includegraphics[width=43mm, height=32mm]{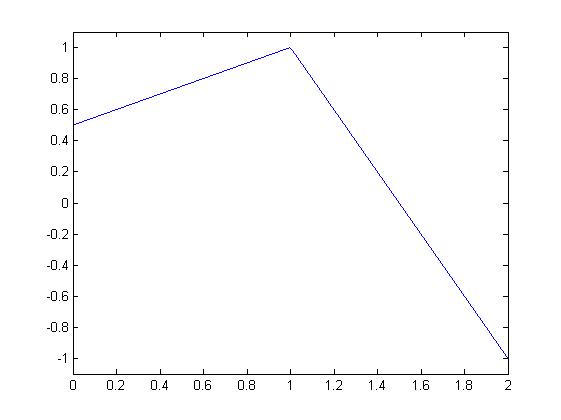}
 % &  \includegraphics[width=22mm, height=20mm]{bilevelsignal_filter2.jpg}
  %{bilevelsignal_convolution_sample.jpg}
%\\
%(a) \ Bi-level signal $x_0(t)$ & (b) \ Sampling kernel $h$  %& Samples $\{x*h(n)\}$
   \end{tabular}
\caption{\small Bi-level signal $x_0$ (left) and sampling kernel $h_0$ (right)
% Plotted on the left is the bilevel signal
%$x(t)=\chi_{[0.6595, 2.1299)}(t) +\chi_{[3.4226,    4.5345)}(t)+\chi_{[6.0301,    7.1877)}(t)+
%\chi_{[8.2906,     9.6402)}(t)+\chi_{[11.0856,   12.5653)}(t)$,
%whose  maximal local rate of innovation is $0.9067$.
%Plotted on the right is the  piecewise polynomial sampling kernel $h(t)=
%=(t+1)/2\chi_{[0,1)}(t)+ (2t-1)\chi_{[1, 2)}$.
%% is plotted in the middle, and
%%  the  samples $\{x_0*h(n)\}_{n=1}^{15}=\{ 0.1966,    0.8132,  -0.1950,    0.6352 ,   0.5499,   -0.2042,  0.8133,    0.2463, $ $   -0.0078,    0.8720,   -0.1360,    0.5621,    0.4605,   -0.2340,  0\}$ is on the right.
}
 \label{bilevelsignal.fig}
\end{figure}
Here  transition positions $t_i^0, 1\le i\le 10$, of the bilevel signal $x_0$ are randomly selected so that
$t_i^0-t_{i-1}^0\in [1.1, 1.9], 2\le i\le 10$. The bilevel signal $x_0$ in (\ref{testsignal}) has $0.8756$ as its maximal local rate of innovation.
%Consider the perturbed bilevel signal $\tilde x_0$
% having transition positions $t_i^0+\delta \epsilon_i, 1\le i\le 10$,
% where $0\le \delta\le 0.2$,
% and noises $\epsilon_i, 1\le i\le 10$, are randomly selected in $[-1, 1]$.
% Plotted in Figure \ref{samplingerror.fig} is
%maximal sampling error $S_\delta:=\sup_{1\le n\le 14} |x_0*h_0(n)-\tilde x_0*h_0(n)|$
%when we apply the sampling procedure $\tilde x_0\longmapsto \{\tilde x_0*h_0(n)\}$
%50 times  for every noise level $\delta\in [0, 0.2]$.
%This confirms the stability conclusion (\ref{error.prop.eq1}) in Theorem \ref{stability.tm} for  the sampling procedure from bilevel signal $x$
%to its samples $\{x*h(n/T)\}$.
%
%%where $T=1$,  $x(t)=\chi_{[0.65,    1.75)}(t)+\chi_{[3.09,    4.56)}(t)+\chi_{[5.90,    7.28)}(t)+
%% \chi_{[8.65,    9.85)}(t)+\chi_{[11.17,  12.36)}(t)$ and .
%%
%  \begin{figure}[hbt]
%\centering
%\begin{tabular}{ccc}
%  %\includegraphics[width=42mm, height=28mm]{samplingerror_average.jpg}
%  &  \includegraphics[width=68mm, height=48mm]{samplingerror_max_50.jpg}
%%\\
%%(a) \ Signal $x_0(t)$ & (b) \ Filter $h$ & Samples $\{x*h(n)\}$
%   \end{tabular}
%\caption{\small Maximal sampling error
%%when innovation positions are randomly perturbed for 20 times. % when performing the sampling
%}
% \label{samplingerror.fig}
%\end{figure}
%
%We take $x_0$ in (\ref{testsignal})
%as the original bilevel signal, and $h_1$ in (\ref{testkernel}) as the sampling kernel.
We  sample the convolution $x_0*h_1$ between $x_0$ and $h_1$ every second, which generates the
sampling vector $Y_0=(x_0*h(1), \ldots, x_0*h(14))$,
 %\begin{eqnarray}
% & & (x_0*h(1), \ldots, x_0*h(14)) \nonumber\\
% & = & (0.4068,    0.9796,   -0.0115,    0.6237,    0.4279,\nonumber\\
%   & & \  -0.0883,    0.9349,    0.1736,
%   -0.0034,    0.9413, \nonumber\\
%   & &   \   0.1023,    0.2800,    0.8135,   -0.2146).\end{eqnarray}
   and then we add bounded random noise to the sampling vector
$$Y_\delta=Y_0+\delta (\epsilon_1, \ldots, \epsilon_{14})$$
with  noise level $\delta\ge 0$, where $\epsilon_i\in [-1, 1], 1\le i\le 14$, are random noises.
We apply the bilevel signal recovery algorithm in Section \ref{recovery.section} with some technical adjustment when
the reconstructed transition position is approximately located at some sampling positions,
and denote  the first ten transition positions of the reconstructed bilevel signal  $x_\delta$  by $t_{1, \delta}, \ldots, t_{10, \delta}$.
Define  maximal error of first ten transition positions by
$$P(\delta)=\max_{1\le i\le 10} |t_{i,\delta}-t_i^0|.$$
%where $t_1^0, \ldots, t_{10}^0$ are transition positions of the bilevel signal $x_0$.
We perform the bilevel signal recovery algorithm  in Section \ref{recovery.section} 50 times for every noise level $\delta\in [0, 0.03]$.
The
 maximal value  of $P(\delta)$ after performing the algorithm 50 times is plotted in Figure \ref{innovationerror.fig} with solid line,
 while the average value of $P(\delta)$ plotted with dashed line.
Notice that    $\max_{1\le n\le 14} |x_0*h_1(n)|=0.9796$.
 Thus the maximal error $P(\delta)$ of transition positions
is less than $10\%$  when   the noise level $\epsilon=\max_{n\ge 1} |\epsilon_n|$
is  at (or below) $2\%$ of the maximal sample value
$\max_{n\ge 1} |x_0*h_0(nT)|$, while some
transition positions could not be recovered  approximately when the noise level is above $3\%$.
This   indicates that our  algorithm to recover
    the bilevel signal  from its noisy samples  is ``reliable"
    only for low level of bounded noises.
    We doubt that it is because of  the instability of the  nonlinear recovery procedure
 in the presence of bounded noises. We will do the detailed noise performance analysis in the coming paper.
    % the noise level is  at (or below) $2\%$.

  \begin{figure}[hbt]
\centering
\begin{tabular}{c}
   \includegraphics[width=88mm, height=48mm]{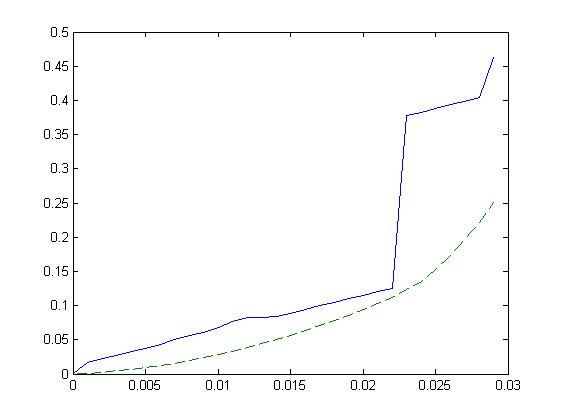}
%\\
%(a) \ Signal $x_0(t)$ & (b) \ Filter $h$ & Samples $\{x*h(n)\}$
   \end{tabular}
\caption{\small Maximal transition position error
%when innovation positions are randomly perturbed for 20 times. % when performing the sampling
}
 \label{innovationerror.fig}
\end{figure}

\section{Conclusion}
In this note, we show that  bilevel causal signals $x$
could be reconstructed from  their samples $x*h(nT), n\ge 1$,
if the  sampling kernel $h$ is  causal and positive on $(0,T)$  and
if the sample rate  is at (or above) the maximal local rate of innovation of the bilevel signal $x$.
We also propose a stable bilevel signal recovery algorithm in the presence of bounded noise
if the number of transition positions of bilevel signals is not large. We remark that the bilevel signal recovery  algorithm
proposed in this note is applicable
when uniform sampling $x*h$ every $T$ second replaced by
nonuniform sampling $\{x*h(s_n)\}$ with sampling density  $\sup_{n\ge 1} |s_{n+1}-s_n|\le T$,
and bilevel causal signal $x=\sum_{i=1}^N \chi_{[t_{2i-1}, t_{2i})}(t)$ with maximal local rate of innovation  $R\le 1/T$ replaced by
box causal signals $x=\sum_{i=1}^N c_i \chi_{[t_{2i-1}, t_{2i})}(t)$ with maximal local rate of innovation  $R\le 1/(2T)$, where
for every $1\le i\le N$,  $c_i$ is height of the box located on the time period $[t_{2i-1}, t_{2i})$.


\section*{Acknowledgment}

The authors would like to thank Professor Ram Mohapatra
 for his help  in the preparation of this note. This work is  supported in part by  the National Science Foundation
(DMS-1109063).



%\bibliographystyle{IEEEtran}

\end{document}